\documentclass[letterpaper]{article} 
\usepackage{aaai23}  
\usepackage{times}  
\usepackage{helvet}  
\usepackage{courier}  
\usepackage[hyphens]{url}  
\usepackage{graphicx} 
\usepackage{natbib}  
\usepackage{caption} 
\usepackage{algorithm}
\usepackage[noend]{algpseudocode}
\usepackage{newfloat}
\usepackage{listings}
\DeclareCaptionStyle{ruled}{labelfont=normalfont,labelsep=colon,strut=off} 
\floatstyle{ruled}
\newfloat{listing}{tb}{lst}{}
\floatname{listing}{Listing}
\usepackage{csquotes}
\usepackage{caption}
\usepackage{subcaption}
\usepackage{amsmath}
\usepackage{amssymb}
\usepackage{amsthm}
\usepackage{mathtools}
\usepackage{optidef}
\usepackage{booktabs}
\usepackage{cleveref}
\usepackage{tikz}
\usetikzlibrary{bayesnet, quotes, decorations.pathreplacing,decorations.pathmorphing}
\usepackage{makecell}

\newtheorem{theorem}{Theorem}

\newtheorem{corollary}{Corollary}

\theoremstyle{definition}

\nocopyright

\SetCiteCommand{\citep}

\DeclareMathOperator*{\argmin}{arg\,min}
\title{Ballot Length in Instant Runoff Voting\footnote{The extended version of a paper appearing at AAAI '23.}}
\author {
    Kiran Tomlinson,\textsuperscript{\rm 1}
    Johan Ugander,\textsuperscript{\rm 2}
    Jon Kleinberg\textsuperscript{\rm 1}
}
\affiliations {
    \textsuperscript{\rm 1}Cornell University\\
    \textsuperscript{\rm 2}Stanford University\\
  kt@cs.cornell.edu, jugander@stanford.edu, kleinberg@cornell.edu
}
\begin{document}

\maketitle
\begin{abstract}
 Instant runoff voting (IRV) is an increasingly-popular alternative to traditional plurality voting in which voters submit rankings over the candidates rather than single votes. In practice, elections using IRV often restrict the \emph{ballot length}, the number of candidates a voter is allowed to rank on their ballot. We theoretically and empirically analyze how ballot length can influence the outcome of an election, given fixed voter preferences. We show that there exist preference profiles over $k$ candidates such that up to $k-1$ different candidates win at different ballot lengths. We derive exact lower bounds on the number of voters required for such profiles and provide a construction matching the lower bound for unrestricted voter preferences. Additionally, we characterize which sequences of winners are possible over ballot lengths and provide explicit profile constructions achieving any feasible winner sequence.
 We also examine how classic preference restrictions influence our results---for instance, single-peakedness makes $k-1$ different winners impossible but still allows at least $\Omega(\sqrt k)$.
Finally, we analyze a collection of 168 real-world elections, where we truncate rankings to simulate shorter ballots. We find that shorter ballots could have changed the outcome in one quarter of these elections. 
Our results highlight ballot length as a consequential degree of freedom in the design of IRV elections.
\end{abstract}

\begin{figure*}[t]
   \begin{subfigure}[b]{0.49\textwidth}
         \centering
         \includegraphics[width=\textwidth]{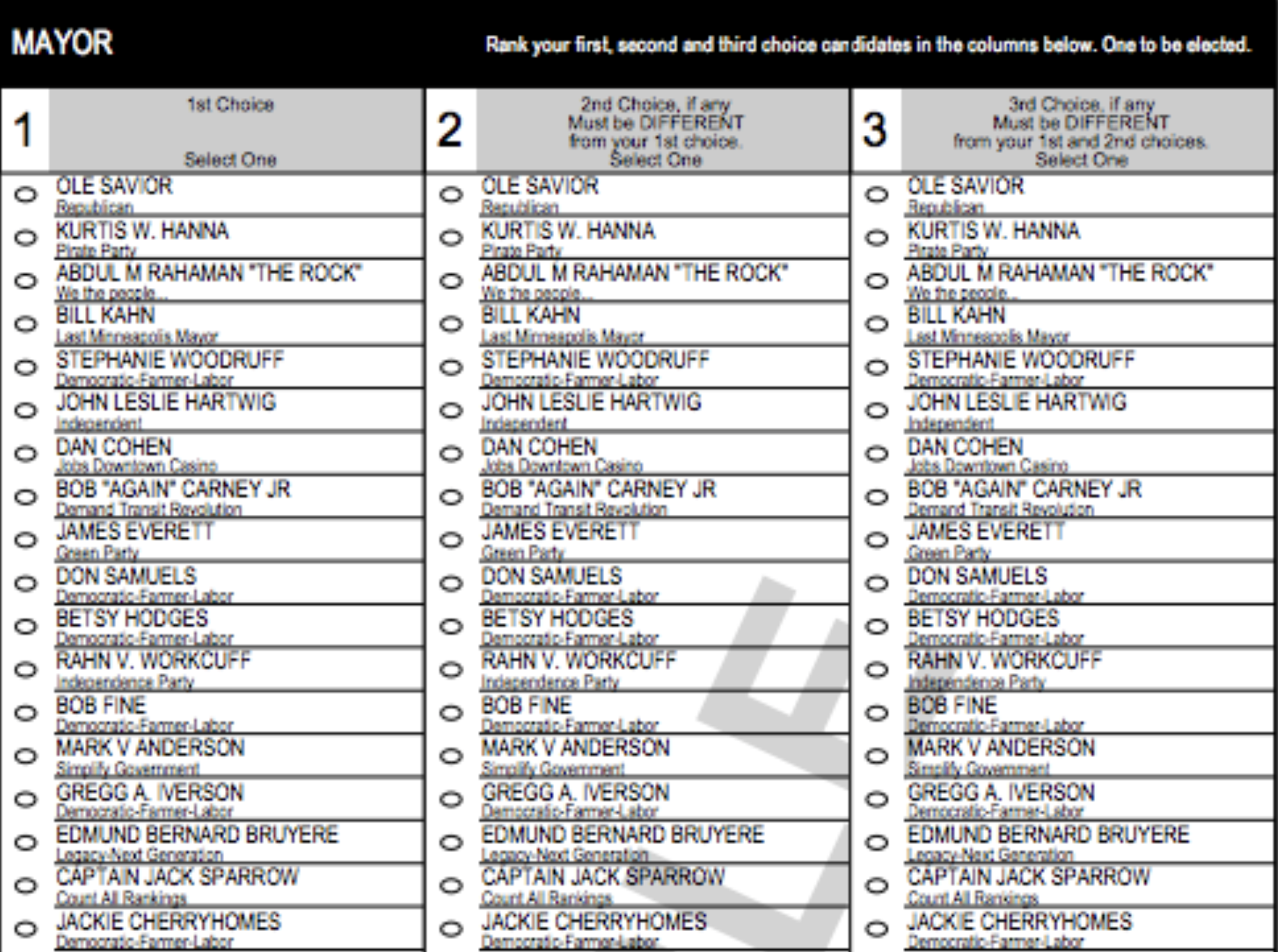}
         \caption{Minneapolis, MN: length 3}
     \end{subfigure}
     \hfill
     \begin{subfigure}[b]{0.49\textwidth}
         \centering
         \includegraphics[width=\textwidth]{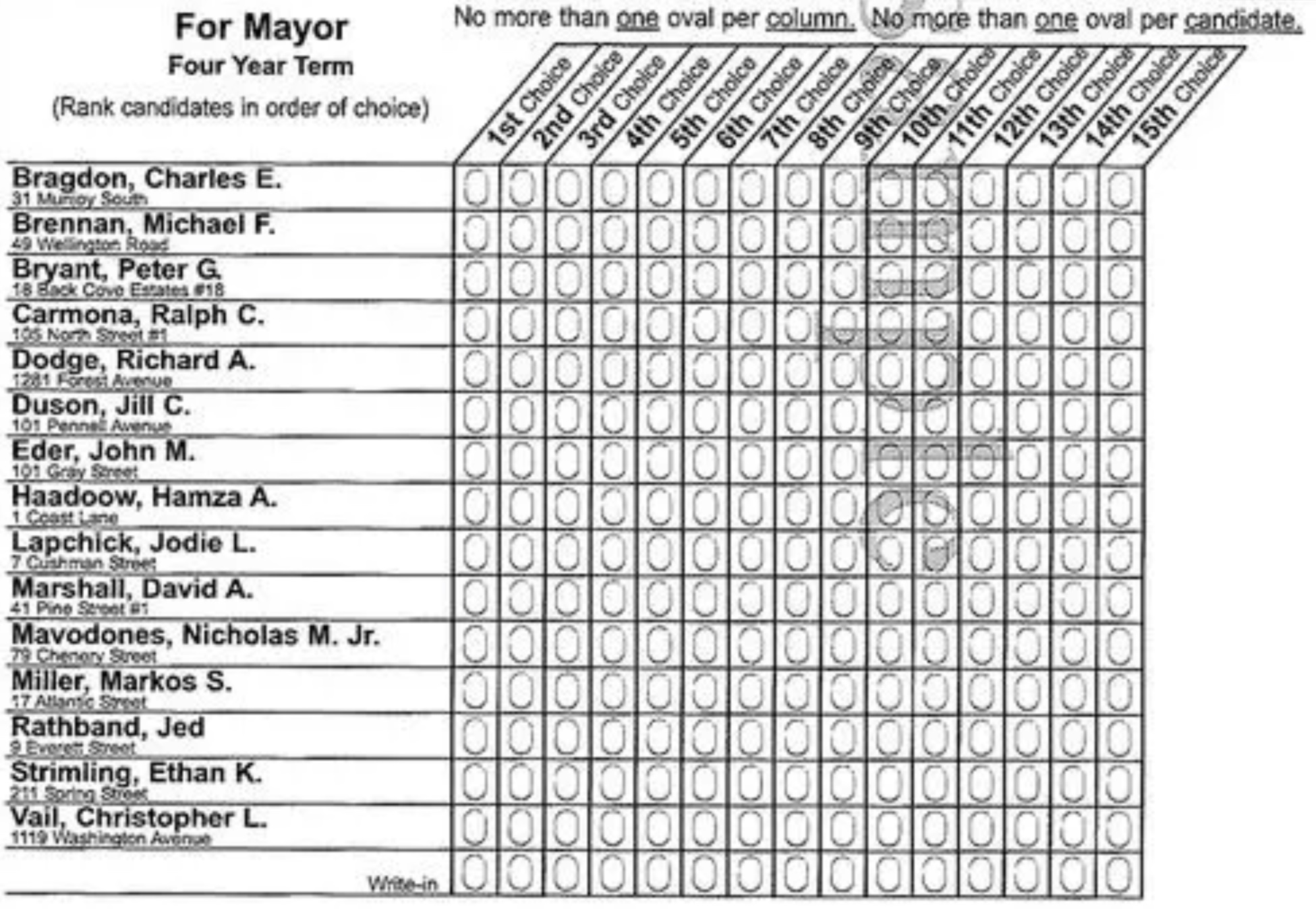}
         \caption{Portland, ME: unlimited length}
     \end{subfigure}\\
     
     \caption{Sample mayoral election ballots from Minneapolis, MN and Portland, ME. Minneapolis ballots allow voters to rank up to three of the candidates, while Portland ballots allow voters to rank all of the candidates.  }\label{fig:ballots}
\end{figure*}

\section{Introduction}
Instant runoff voting (IRV) has grown in popularity over the last two decades as an alternative to plurality voting for governmental and organizational elections. Also referred to as ranked choice voting (RCV), single transferrable vote (STV), alternative vote, preferential voting, or the Hare method, IRV allows voters to submit rankings over the candidates rather than voting for a single option. 
IRV determines a winner from these rankings by repeatedly eliminating the candidate who has the fewest ballots ranking them first; the ballots that listed this eliminated candidate first have their votes reallocated to the next candidate on their list.
This process continues, repeatedly eliminating candidates, until only one is left---the winner.

Proponents of IRV argue that it allows voters to report their full preferences, mitigates vote-splitting when similar candidates run, encourages civility in campaigning, and saves money compared to holding separate runoff elections~\cite{fairvote, lewyn2012two}. Many local elections in the United States use IRV, including in Minneapolis, San Fransisco, Oakland, Santa Fe, and New York City,
 as well as statewide elections in Maine and Alaska. IRV is also used in other countries, including Australia and Ireland. 

However, IRV has vocal opponents who believe it to be too confusing for voters~\cite{langan2004instant,wsjmess}, leading to outright bans on the use of IRV in Florida~\cite{floridaban} and Tennessee~\cite{tennesseeban}. One particular issue critics point to is the complexity of a ballot that asks voters to rank every candidate, especially when the number of candidates is large. One official tasked with running Utah's first IRV election raised this as her primary concern after the election:
\begin{displayquote}[\citealp{swensenquote}, Salt Lake County Clerk]
 My concerns with the current RCV law are that we would recommend the number of rankings be limited to three or five instead of an unlimited number based on the number of candidates. So although you can list as many candidates as file on the ballot, I think it is a bit confusing to voters [...] For instance, in Minneapolis they rank three. In St.\ Paul, they rank five. They don’t usually have them rank as many candidates as there are.
\end{displayquote}
Indeed, many municipalities have different numbers of ranking slots on their IRV ballots, what we call \emph{ballot length}: Oakland uses three, Alaska four, and New York City five. The count goes on: ballot length six would have been mandated by the failed 2019 Ranked Choice Voting Act proposing IRV for US Congressional elections~\cite{rcvact}. In Maine, voters can rank all of the candidates---even if there are 15 of them.
In fact, plurality voting can be viewed as IRV with ballot length one: losing candidates are repeatedly ``eliminated'' (without redistribution) until the candidate with a plurality is declared the winner.

While making ballots shorter does make them simpler, it also strays from a goal of IRV: allowing voters to express their complete preferences over the candidates. Critics of IRV also raise concerns about \emph{ballot exhaustion} during the IRV algorithm, where all candidates ranked by a voter have been eliminated and that vote no longer contributes to subsequent tallies~\cite{burnett2015ballot}.\footnote{In plurality, any vote not cast for the winner is ``exhausted.''} Ballot length is therefore subject to competing desires: shorter ballots are easier to fill out and simpler to print, but less informative about voter preferences. 

Despite the apparent trade-offs involved in ballot length, 
there has been very little investigation of how these trade-offs might work.
As noted above, plurality voting can be seen as IRV with ballot
length one, and so the fact that plurality and IRV can produce 
different outcomes already indicates that ballot
length can have important consequences.
But aside from early work looking at simulations and a few real-world elections \cite{kilgour2020prevalence,ayadi2019single} we
do not have much insight into the consequences of ballot length more generally.
Perhaps, for example, there are underlying structural properties to be
discovered that constrain how many winners are possible as we vary the ballot length.
Or perhaps ``anything goes,'' and if we specify which candidate
we'd like to see win at each possible ballot length, we can construct a fixed
set of rankings that produce each desired winner 
at the corresponding length.

\paragraph{Overview of Results.}
In this paper, we show that the effect of ballot length 
essentially behaves like the latter extreme, 
where almost every sequence of outcomes is possible.
In particular, we prove that 
modulo a simple feasibility constraint, it
is possible to pick any sequence of candidates (with repetitions allowed), and to have this be the sequence of winners
at ballot lengths $1, 2, 3, ... $.
For example, there are voter preferences such that one candidate wins if the election is run with odd ballot length and another wins with even ballot length. We make a central assumption that voters have fixed ideal rankings and report as long a prefix of their ideal ranking as the ballot allows. Given $k$ candidates, we show that up to $k-1$ of them can win as the ballot length varies from $1, \dots, k-1$ and voter preferences remain fixed. Moreover, we establish exact matching lower bounds on the number of voters required to produce $k-1$ distinct winners. 

We also consider how these results are affected if we make
standard modeling assumptions about voters.
If we model voters abstractly as exhibiting 
single-peaked or single-crossing preferences, we prove that $k-1$ distinct winners across ballot lengths cannot be achieved. 
We also consider voters who rank candidates according to a shared
one-dimensional ideological spectrum; since such voters are both
single-peaked and single-crossing, there cannot be $k-1$ distinct winners
in these cases.
We find through simulation that in this one-dimensional case, 
ballot lengths above $k/2$ almost always produce the same winner as full IRV ballots.

Finally, we use data from 168 real-world elections from PrefLib~\cite{mattei2013preflib} (most of them originally conducted using IRV), and we find that 
different winners across ballot lengths is a phenomenon that occurs commonly: in 25\% of the PrefLib elections at least two different candidates win as the ballot length is varied by truncation. However, truly pathological cases with $k-1$ winners appear to be extremely rare: we observe at most three distinct winners across ballot lengths, and that occurs only once in the 168 PrefLib elections. But even with these real-world voter preferences, more than three winners can occur; by resampling ballots in the PrefLib elections, we observe cases with four, five, and even six different winners across ballot lengths. We note that one third of the elections initially used ballot length of at most four, where it is impossible to have more than three different winners across ballot lengths. Our code and data are available at \url{https://github.com/tomlinsonk/irv-ballot-length}.

\section{Related work}
There has been considerable work on what happens when individual voters choose not to rank all the candidates---a practice sometimes called \emph{voluntary truncation}---in contrast with \emph{forced truncation} (i.e., ballot length restrictions)~\cite{kilgour2020prevalence}. In many voting systems including IRV, election outcomes can change dramatically as voters independently choose to rank more or fewer candidates~\cite{saari1988problem}. This matter has been studied from a computational angle as the \emph{possible winners} problem, which asks, given a collection of partial ballots, which candidates could become winners as those ballots are filled out~\cite{konczak2005voting,chevaleyre2010possible, baumeister2012campaigns,xia2011determining,ayadi2019single}. There is  also a wide array of research on how partial ballots can be used for strategic voting and campaigning~\cite{baumeister2012campaigns,narodytska2014computational, menon2017computational,kamwa2022scoring,fishburn1984manipulability}. On the empirical side, voluntary truncation is a concern since it can lead to ballot exhaustion~\cite{burnett2015ballot}. In political science, voluntary truncation is also referred to as \emph{under-voting}~\cite{neely2008whose}. Several studies have asked whether different demographic groups are more likely to under-vote and how this could have a disenfranchising effect~\cite{neely2008whose,coll2021demographic,hoffman2021proportionality}. There has also been research on ``over-voting'' in IRV, which refers to ranking a single candidate in more than one position (e.g., first and second), especially its correlation with underrepresented voting populations~\cite{neely2008whose,neely2015overvoting}. 

 In contrast, we investigate what happens when all voter preferences are truncated as a result of ballot length. That is, we focus on a question of election design rather than on voter choice. In this direction, \Citeauthor{ayadi2019single}~\citeyearpar{ayadi2019single} investigated how often IRV with short ballots produces the full-ballot winner in the Mallows model and in five PrefLib elections. However, all five PrefLib elections they studied produced the full-ballot winner at all ballot lengths---in analyzing  a larger collection of 168 PrefLib elections, we find multiple winners across ballot lengths in 25\% of them. \Citeauthor{ayadi2019single} also examined several other interesting facets of IRV ballot length, including a low-communication IRV protocol (a form of online, per-voter ballot length customization) and the complexity of the possible winners problem under truncated ballots. The issue of ballot length in IRV was also touched on by \Citeauthor{kilgour2020prevalence}~\citeyearpar{kilgour2020prevalence}, who examined its effect in simulation for $k = 4, 5,$ and $6$ candidates, where they found up to $k - 2$ distinct winners across ballot lengths. We prove that in fact $k-1$ winners are possible \emph{for all $k\ge 3$}. Ballot length has been considered in contexts other than IRV---for instance, research on the Boston school choice mechanism found that limiting the number of schools parents could rank to five resulted in undesirable strategic behavior~\cite{abdulkadiroglu2006changing}. There has also been research on ballot length in approval voting from a learning theory angle, seeking to recover a population's preferences efficiently~\cite{garg2019your}.

\section{Preliminaries}

\begin{figure}
  
  \centering
  \begin{tikzpicture}[x=3.6mm,y=4mm]
    \newcommand{\da}{-1}
    \newcommand{\db}{12.5}
    

    \node[align=center] (a) at (3, 5) {\emph{voter count}};
    \node[align=center] (c) at (3, -0.5) {\emph{ballot type}};

    \draw[draw=black,fill=white] (\da+-0.5,-0.5) rectangle (\da+0.5,3.5);
    \node  at (\da+0, 0) {B};
    \node (d) at (\da+0, 1) {C};
    \node at (\da+0, 2) {D};   
    \node at (\da+0, 3) {A};
    \node at (\da+0, 4) {$2$};
    
    \node at (\da+1.5, 3) {A};
    \node at (\da+1.5, 4) {5};
    \draw[draw=black] (\da+1,2.5) rectangle (\da+2,3.5);
 
    \node at (\da+3, 1) {A};
    \node at (\da+3, 2) {D};
    \node at (\da+3, 3) {B};
    \node at (\da+3, 4) {6};
    \draw[draw=black] (\da+2.5,0.5) rectangle (\da+3.5,3.5);
    
    \node at (\da+4.5, 3) {C};
    \node at (\da+4.5, 4) {6};
    \draw[draw=black] (\da+4,2.5) rectangle (\da+5,3.5);
    
    \node at (\da+6, 2) {B};   
    \node at (\da+6, 3) {D};
    \node at (\da+6, 4) {3};
    \draw[draw=black] (\da+5.5,1.5) rectangle (\da+6.5,3.5);
   
    \node at (\da+7.5, 2) {C};   
    \node at (\da+7.5, 3) {D};
    \node at (\da+7.5, 4) {2};
    \draw[draw=black] (\da+7,1.5) rectangle (\da+8,3.5);

    \draw[-stealth] (c) -- (d);

    \draw[-stealth] (8, 2) -- (11, 2) node[midway,above] {$h = 2$};

    \node at (\db+0, 2) {D};   
    \node at (\db+0, 3) {A};
    \node at (\db+0, 4) {2};
    \draw[draw=black] (\db+-0.5,1.5) -- (\db+-0.5, 3.5) -- (\db+0.5,3.5) -- (\db+0.5,1.5);
    \draw[decorate,decoration={zigzag,segment length=0.75mm,amplitude=0.2mm}] (\db+-0.5,1.5) -- (\db+0.5,1.5);
    
    \node at (\db+1.5, 3) {A};
    \node at (\db+1.5, 4) {5};
    \draw[draw=black] (\db+1,2.5) rectangle (\db+2,3.5);
 
    \node at (\db+3, 2) {D};
    \node at (\db+3, 3) {B};
    \node at (\db+3, 4) {6};
    \draw[draw=black] (\db+2.5,1.5) -- (\db+2.5, 3.5) -- (\db+3.5,3.5) -- (\db+3.5,1.5);
        \draw[decorate,decoration={zigzag,segment length=0.75mm,amplitude=0.2mm}] (\db+2.5,1.5) -- (\db+3.5,1.5);
        
    \node at (\db+4.5, 3) {C};
    \node at (\db+4.5, 4) {6};
    \draw[draw=black] (\db+4,2.5) rectangle (\db+5,3.5);
    
    \node at (\db+6, 2) {B};   
    \node at (\db+6, 3) {D};
    \node at (\db+6, 4) {3};
    \draw[draw=black] (\db+5.5,1.5) rectangle (\db+6.5,3.5);
   
    \node at (\db+7.5, 2) {C};   
    \node at (\db+7.5, 3) {D};
    \node at (\db+7.5, 4) {2};
    \draw[draw=black] (\db+7,1.5) rectangle (\db+8,3.5);

  \end{tikzpicture}
  \caption{On the left, an example profile with $k=4$ candidates A, B, C, D and $n= 24$ voters of 6 types with partial ballots. Ballots are listed top-down, with the number of voters of each type above each ballot. On the right, the profile is truncated to ballot length $h=2$.}\label{fig:ex}
\end{figure}
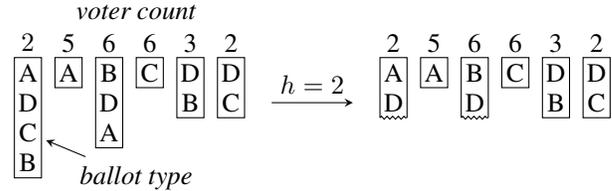

An IRV election consists of $k$ candidates labeled $1, \dots, k$ and $n$ voters. Each voter $j$ has a preference ordering over a subset of the candidates denoted by the ordered subset $\pi_j$, which we refer to as a \emph{ballot}. At any point down the ballot, $\pi_j$ can terminate, at which point the voter is indifferent over the remaining options. If $\pi_j$ includes all candidates, we call it \emph{full}, otherwise we call it \emph{partial}. We call a collection of ballots a \emph{profile}. Unless otherwise specified, a profile may contain partial ballots.\footnote{All 168 elections in the PrefLib data have partial ballots.} If multiple voters have identical ballots, we say they are of the same \emph{type}. Given a profile, IRV proceeds by eliminating the candidate with the fewest ballots ranking them first and removing them from all ballots. Ballots that have all their candidates eliminated are \emph{exhausted}. Eliminations continue until only one candidate remains, who is declared the winner (equivalently, one can terminate when one candidate has the majority of votes from non-exhausted ballots). Ties can be broken as desired (for instance, by coin-flip), although they are unlikely in large elections. 

In many real-world elections, the number of candidates a voter can rank is limited to $h<k$, which we call the \emph{ballot length}.  We assume that if the ballot length is $h$, voters submit the length $h$ prefix $\pi_{j}(1, \dots, h)$ of their ideal ballot $\pi_j$.  Voters who would have submitted a ranking listing $h$ or fewer candidates are unaffected. Thus, we say that ballots are \emph{truncated} to the ballot length $h$. See \Cref{fig:ex} for an example of a profile with partial ballots truncated to $h=2$.
Note that there is no difference between running IRV with ballot length $k$ and $k-1$, since only one candidate remains after the $(k-1)$th elimination.

The main question we focus on is how ballot length affects an election. For instance, how many different candidates can win as the ballot length varies for a fixed profile? In order to address this question, we make some assumptions about the lack of consequential ties, since in trivial cases such as zero voters, any candidate can win depending on tie-breaks. We say that a profile is \emph{consequential-tie-free} if tie-breaks do not affect the winner under any ballot length $h$.  We say it is \emph{elimination-tie-free} if a tie for last place never occurs when running IRV for any ballot length $h$. Finally, we say it is \emph{tie-free} if no two candidates ever have tied vote counts when running IRV at any ballot length $h$. We note that the problem of determining if a given candidate could win under some tie-breaking sequence is known to be NP-complete~\cite{conitzer2009preference}.

\section{Worst-case analysis of ballot truncation}
We say a profile has $c$ \emph{truncation winners} if $c$ different candidates can win depending on the ballot length. Previous simulation work found up to $k-2$ truncation winners for $k = 4, 5,$ and $6$~\cite{kilgour2020prevalence}. One of our main results is that up to $k-1$ truncation winners are possible for any $k$. We note that it is impossible to have all $k$ candidates win under different ballot lengths, since lengths $k$ and $k-1$ behave the same way. 

First, we establish an exact lower bound on the number of voters required in order to achieve $k-1$ truncation winners in consequential-tie-free profiles.
Our voter lower bound is based on the observation that the winner at $h=1$ (the plurality winner) must be eliminated second under ballot lengths $\ge 2$ for $k-1$ truncation winners to occur. In order for the plurality winner to be eliminated second, the first elimination must redistribute enough votes for every other candidate to overtake the plurality winner.

\begin{theorem}\label{thm:tie-bound}
  For any $k > 3$, a consequential-tie-free profile must contain at least $2k^2-2k$ voters in order to produce $k-1$ truncation winners. For $k=3$, the lower bound is $k^2 = 9$. 
\end{theorem}

\begin{proof}
 Suppose we have a consequential-tie-free profile with $k-1$ truncation winners.  Then $k-1$ of the candidates each have a unique ballot length in $1, \dots, k-1$ at which they win.  Label these candidates $1, \dots, k-1$ according to their winning ballot length. The candidate not in those $k-1$ winners, call them candidate $k$, must have at least 1 fewer first-place vote than any other candidate (otherwise one of the winners could be eliminated first after a tie-break, preventing them from winning at their ballot length). Now consider the winner under ballot length 1, namely candidate 1. In order for candidate 1 to be the unambiguous plurality winner, they must have at least one more vote than every other candidate. Next, consider who is eliminated second. It has to be candidate 1: if any other candidate $i \ne 1$ can be eliminated second, then they will not be able to win at their designated ballot length $h > 1$. In order for candidate 1 to be eliminated second, they must be in unambiguous last place after candidate $k$'s ballots are redistributed. This means at least 2 of those ballots need to go to each of candidates $2, \dots, k-1$ (who are currently trailing candidate 1 by 1 vote). Finally, if $k>3$, the candidate who wins at ballot length 2  (candidate 2) must be unambiguously in the lead over $3, \dots, k-1$ after redistributing $k$'s ballots. Either they had more initial ballots than $3, \dots, k-1$ (but this would require at least one more ballot from candidate $k$ to help those lower candidates overtake 1) or they got a single extra ballot from candidate $k$. To summarize the constraints:
\begin{enumerate}
  \item candidates $1, \dots, k-1$ have at least one more first-place vote than candidate $k$,
  \item candidate $1$ has at least one more first-place vote than any other candidate, and
  \item candidate $k$ has enough first-place votes to redistribute at least two each to $2, \dots, k-1$ (plus at least one more if $k > 3$). 
\end{enumerate}
For $k>3$, the total number of ballots ranking $k$ first is thus at least $2(k-2) + 1$, by constraint 3. Each of candidates $2, \dots, k-1$ must then have at least $2(k-2) + 2$ first-place ballots by constraint 1. Finally, candidate $1$  must have at least $2(k-2)+3$ first-place ballots by constraint 2. The minimum number of ballots is thus $2(k-2) + 1 + (k-2)(2(k-2) + 2) + 2(k-2)+3 = 2k^2 - 2k$. 

For $k=3$, constraint 3 only requires $2(k-2) = 2$ first-place votes for candidate 3. Candidates 2 and 1 must then have 3 and 4 first-place votes by constraints 1 and 2, for a total of $2 + 3 + 4 = 9 = k^2$.

\end{proof}

Our main theoretical result is a construction matching this lower bound, showing that $k-1$ truncation winners can occur for any $k \ge 3$. Our construction can not only produce $k-1$ truncation winners, but \emph{any} sequence of winners over ballot lengths $1, \dots, k-1$, provided that a candidate has not yet been eliminated.

\begin{theorem}\label{thm:tie-construction}
Let there be $k> 3$ candidates, labelled $1, \dots, k$ in their full-ballot IRV elimination order.
 Fix any sequence of candidates $w_1, \dots, w_{k-1}$ such that $w_h \in \{h+1, \dots, k\}$ for all $h\in [k-1]$. There exists a consequential-tie-free profile with $2k^2 - 2k$ partial ballots whose sequence of truncated IRV winners from $h=1, \dots, k-1$ is $w_1, \dots, w_{k-1}$. For $k=3$, such a profile exists with $9$ ballots. Any sequence where $w_h \le h$ for some $h \in [k-1]$ is impossible to realize as the sequence of truncated IRV winners for any consequential-tie-free profile. 
\end{theorem}

\begin{proof}
  First, if we have a sequence with $w_h \le h$ for some $h$, then this means the winner at ballot length $h$ is eliminated $h$th or sooner under ballot lengths $\ge h$. This is impossible, since they would be eliminated before they win at length $h$.

  Now suppose we have some valid sequence $w_1, \dots, w_{k-1}$ such that $w_h \in \{h+1, \dots, k\}$ for $h \in [k-1]$. First, assign $2(k-2) + 1$ ballots to each candidate listing them first. Give candidate $w_1$ an extra 2 ballots and the other candidates (except candidate $1$) an extra $1$ ballot each. This is a total of $2k(k-2) + k + (k-2) + 2 = 2k^2 - 2k$ ballots. We now fill out the ballots initially assigned to each candidate, using $S_i$ to denote the set of ballots ranking $i$ first.
  
Except for $i=1$, all ballots in $S_i$ rank candidates $1, \dots, i-1$ in positions $2, \dots, i$. For all $i$, two ballots in $S_i$ rank $\ell$ in position $i+1$ for each $\ell = i+2, \dots, k$ except $w_i$. If $w_{i} \ne i+1$, one ballot in $S_i$ ranks $w_i$ in position $i+1$. Finally, one extra ballot in $S_i$ ranks $w_{i+1}$ in position $i+1$. This requires at most $2(k-2)+1$ ballots, which is covered by the $\ge 2(k-2)+1$ ballots in $S_i$. All ballots in $S_i$ then terminate after their last specified entry. Notice that when $i$ is eliminated, the effect of their redistributed votes is to put the new winner $w_{i+1}$ in the lead and the new loser $i+1$ in last, assuming the last winner $w_i$ was in the lead by a single vote after $i$ is eliminated.

  We now show that if ballots are truncated to length $h < k$, then candidate $w_h$ wins under IRV. First, if we truncate ballots to length $1$, candidate $w_1$ wins: they have 2 more first place votes than candidate $1$ and 1 more than every other candidate. Thus, candidate $1$ will be eliminated (with no redistribution due to the length-1 ballots), followed by the others in some order based on tie-breaking, making candidate 1 win.
  
     Now suppose we truncate to length $h$ ($2 \le h < k$). Candidate $1$ is eliminated first and their second place votes cause candidates $3, \dots, k$ to overtake candidate $2$, with candidate $w_2$ taking the lead by 1 vote. If $h = 2$, then all remaining ballots only have one candidate listed (since the second place votes for ballots assigned to candidate $\ell > 1$ are all for candidate $1$, who is eliminated). Thus candidate $w_2$ wins after eliminating candidate $2$ and then $3, \dots, k \setminus w_2$ in some order. For $h> 2$, we'll prove inductively that for $2 \le \ell < h$, the $\ell$th candidate eliminated is candidate $\ell$, which causes candidate $w_{\ell+1}$ to take the lead by one vote and candidate $\ell+1$ drop to last place by one vote. 
       
     \underline{Base case} ($\ell = 2$): As we saw, the 2nd candidate eliminated is candidate 2. Since $h > 2$, ballots assigned to candidate 2 are not yet exhausted: two go to each of candidates $4, \dots, k$ (except $w_2$); $w_2$ gets one if $w_2 \ne 3$ and zero otherwise; and $w_3$ gets one extra ballot. Since candidate $w_2$ was only in the lead by one vote, this causes the new leader to be candidate $w_3$ and candidate 3 to drop to last place, as claimed. 
     
     \underline{Inductive case} ($2 < \ell < h$): by inductive hypothesis, candidates $2, \dots, \ell -1$ have been eliminated (plus candidate $1$, the first to go), candidate $w_{\ell}$ is currently in the lead, and candidate $\ell$ is in last place. Thus, candidate $\ell$ is the $\ell$th to be eliminated. By construction, the candidates ranked in positions $2, \dots, \ell$ on the ballots initially assigned to $\ell$ (namely, candidates $1, \dots, \ell - 1$) have been eliminated. Additionally, all ballots that were redistributed to $\ell$ are now exhausted. Since $\ell < h$, there are still remaining places on the truncated ballot. Ballots currently assigned to $\ell$ are distributed as follows: two go to each of $\ell+1, \dots, k$ (except $w_\ell$); $w_\ell$ gets one if $w_\ell \ne \ell + 1$ and zero otherwise; and $w_{\ell+1}$ gets one extra ballot. This causes candidate $w_{\ell+1}$ to take the lead by one vote and candidate $\ell$ to drop to last place behind $\ell +2, \dots, k-1$, as claimed.
     
  Once candidate $w_h$ is in the lead, candidates $1, \dots, h-1$ have been eliminated, and candidate $h$ is in last place, all the ballots only list the candidate to which they are currently assigned (since the candidates ranked up to position $h$ on their ballots have been eliminated). Thus, $h$ will be eliminated, followed by $h+1, \dots, k$ (except $w_h$) in some order, making the winner candidate $w_h$, as desired.
\end{proof}

The idea behind the construction is to maintain a tie for second place among all candidates but two: the candidate about to be eliminated, in last, and the candidate next in the winner sequence, in first. Each elimination redistributes ballots to move the next candidates into first and last place. By carefully designing ballots, they become exhausted at just the right moment to freeze the order once we reach step $h$ of IRV, causing the candidate currently in first to win. The example in \Cref{fig:ex} uses this construction for $k=4$ to achieve different winners at ballot lengths $1, 2, 3$ (namely, A, B, C). Note that the full-ballot elimination order labeling of candidates A, B, C, D is 2, 3, 4, 1, which makes the truncation winner sequence 2, 3, 4 feasible. In contrast, the sequence 2, 2, 4 would not be feasible since the candidate eliminated second under full ballots cannot win at ballot length 2. Intuitively, a winner sequence with elimination order labeling is feasible if it is element-wise at least $2, 3, \dots, k$.

\subsection{Restrictions on profiles}\label{sec:profile-restrictions}
Since IRV can behave very erratically across ballot lengths for general profiles, we might hope that imposing restrictions on the space of profiles makes IRV more well-behaved. We consider three classic profile restrictions from voting theory, single-peaked~\cite{black1948rationale,arrow1951social}, single-crossing~\cite{gans1996majority}, and 1-Euclidean preferences (see~\cite{elkind2022restrictions} for a survey of preference restrictions).  
 A profile is \emph{single-peaked} if there exists an order $<$ over the candidates such that, for every ballot $b$ ranking $i$ first, if $j < k < i$ or $i < k < j$, then $j$ is not ranked above $k$ in $b$.  A profile is \emph{single-crossing} if there exists an ordering $L$ of the ballots such that for every ordered pair of candidates $(i, j)$, the set of ballots ranking $i$ above $j$ forms an interval of $L$. Finally, a profile is \emph{1-Euclidean} if there exist embeddings of the voters and candidates in $[0, 1]$ such that if voter $b$ is closer to candidate $i$ than to candidate $j$, then voter $b$ ranks $i$ above $j$.

Intuitively, single-peaked profiles arise when there is a political axis arranging candidates from left to right and voters prefer candidates closer to their ideal point on the axis (each voter can have their own ideal point). Single-crossing preferences arise when voters are arranged on an ideological axis and each candidate is most appealing to voters at a certain point on this axis. While the definitions appear similar, neither condition implies the other. 1-Euclidean profiles are both single-peaked and single-crossing---but there are profiles that are both single-peaked and single-crossing, but not 1-Euclidean~\cite{elkind2014characterization}. 

In contrast to general profiles, where $k-1$ truncation winners can occur, we show that such cases are impossible under either single-peaked or single-crossing preferences (and therefore 1-Euclidean profiles).

\begin{theorem}\label{thm:single-peaked-bound}
  With $k \ge 5$ candidates, no consequential-tie-free single-peaked profile has $k-1$ truncation winners.  
\end{theorem}

\begin{proof}
Suppose for a contradiction that a single-peaked profile has $k-1$ truncation winners ($k \ge 5$). We know the candidate eliminated first cannot win under any ballot length. In order for the candidate eliminated second ($h\ge 2$) to win at some ballot length, it must be at $h=1$---i.e., the plurality winner must be eliminated second under $h\ge 2$. Thus, they must be overtaken by at least three candidates (for $k\ge 5$) when the first eliminated candidate $X$'s ballots are redistributed. But the second place on ballots listing $X$ first can only be the candidate to the left or right of $X$ in the single-peaked ordering, making this impossible.
\end{proof}

\begin{theorem}\label{thm:single-crossing-bound}

  With $k \ge 5$ candidates, no consequential-tie-free single-crossing profile can result in $k-1$ truncation winners.  
\end{theorem}

\begin{proof}
As in the proof of \Cref{thm:single-peaked-bound}, we'll show that the first candidate eliminated, $X$, can only redistribute ballots to two candidates. Suppose for a contradiction that they redistribute ballots to at least three candidates. Call these candidates $A$, $B$, and $C$ in the order in which they first appear as second choices in the ballots ranking $X$ first   in the single-crossing order $L$. By the single-crossing property, all ballots to the left of ballots starting $X, A$ must rank $A$ above $B$, since a ballot to its right ranks $B$ above $A$, namely those starting $X, B$. Moreover, all ballots to the right of ballots starting $X, C$ must rank $C$ above $B$ by symmetric reasoning. But this means $B$ cannot have any ballots ranking them first, contradicting that $X$ (who does have ballots ranking them first) is the first eliminated. See below for a visual depiction of this argument:
\begin{center}
  \begin{tikzpicture}
    \draw[stealth-stealth] (0, 0) -- (6, 0);
    \node (a) at (2, 1) {
      $\begin{bmatrix}
        X\\
        A\\
        \vdots
      \end{bmatrix}$};
    \node (b) at (3, 1) {
      $\begin{bmatrix}
        X\\
        B\\
        \vdots
      \end{bmatrix}$};
    \node (b) at (4, 1) {
      $\begin{bmatrix}
        X\\
        C\\
        \vdots
      \end{bmatrix}$};
      
      \draw[-stealth] (4, -1) -- (6, -1);
      \draw[stealth-] (0, -1) -- (2, -1);
      
      \draw[|-|] (1.7, -.3) -- (4.3, -.3);
      \node at (3, -.6) {$X$ ranked over $B$};

      \node at (1, -1.3) {$A$ ranked over $B$};
      \node at (5, -1.3) {$C$ ranked over $B$};

      \node at (5.8, -0.3) {$L$};

  \end{tikzpicture}
\end{center}
\end{proof}

Although the upper bound on truncation winners is strictly lower for single-peaked profiles than for general profiles, the number of achievable truncation winners still grows with $k$. In particular, we can show that $\Omega(\sqrt{k})$ truncation winners are possible in a consequential-tie-free single-peaked profile with $\Theta(k)$ voters. 
\begin{theorem}\label{thm:single-peaked-construction}
  With $k = \kappa(\kappa+1)/2$ candidates ($\kappa\ge 3$), there is a single-peaked consequential-tie-free profile with $3\kappa(\kappa+1)/2$ partial ballots that results in $\kappa$ distinct truncation winners. 
\end{theorem}

\begin{proof}
  Call candidates $1, \dots, \kappa$ the \emph{winners}. Each winner $i>1$ has $i - 1$ \emph{filler} candidates $f^{i}_1, \dots, f^{i}_{i-1}$ associated with it. The single-peaked axis has winners in the order $1, \dots, \kappa$, with $i$'s fillers between $i$ and $i-1$. That is, the full axis is $1, f^{2}_1, 2, f^{3}_1, f^3_2, 3, f^{4}_1, f^4_2, f^4_3, 4, \dots, f^\kappa_{\kappa-1}, \kappa$. We will fill out ballots so that $i$ wins at ballot length $i$, while maintaining single-peakedness.
  
  Every winner has $\kappa+1$ ballots listing them first and winner $1$ has an additional single ballot. These ballots then terminate. Each candidate's first filler $f^i_1$ has $i$ ballots that list candidates $f^{i}_1, \dots, f^i_{i-1}, i$ in positions $1, \dots, i$ and then terminate. All other fillers have zero ballots listing them first. 
  
Consider what happens at ballot length $h \le \kappa$. If $h=1$, candidate $1$ wins by one vote. For $1< h \le \kappa$, all fillers with zero ballots are eliminated first in some order. Then, the first fillers are eliminated in the order $f^1_1, f^2_1, \dots, f^\kappa_1$. Only fillers $f^i_1$ with $i \le h$ are able to reallocate votes, since ballots for listing $f^j_1$ ($j>h$) first are exhausted after $f^j_1$'s elimination. The first-place vote counts after all fillers are eliminated are thus $\kappa+2$ for winner 1, $\kappa+1+i$ for winners $2, \dots, h$ and $\kappa+1$ for winners $h+1, \dots, \kappa$. With no more reallocations taking place, candidate $h$ wins. For $h> \kappa$, candidate $\kappa$ still wins. This construction therefore results in $\kappa$ distinct truncation winners.

The total number of candidates in this construction is $\kappa + \sum_{i=2}^{\kappa} (i-1) = \kappa(\kappa +1) / 2$. The total number of voters is $\kappa(\kappa + 1) + 1 + \sum_{i = 2}^\kappa i = 3\kappa(\kappa+1)/2$, as claimed.
\end{proof}

The exact upper bound on the number of truncation winners for single-peaked (and single-crossing) preferences remains an open question---it could be as large as $k-2$. Additionally, we do not know a non-trivial lower bound on the number of achievable truncation winners for single-crossing or 1-Euclidean profiles. 

\subsection{Restrictions on ties}
Since our main theorem allows ties (albeit only ties that do not affect the winners), one might be concerned that the large number of truncation winners is a byproduct of these ties. In the following results, we show that even if no vote counts are ever tied, there can still be arbitrary truncation winner sequences. We can therefore get any feasible winner sequence regardless of the tiebreaking rule.  As before, we start by establishing lower bounds on the number of voters required for $k-1$ truncation winners and then provide a matching construction for tie-free profiles achieving any truncation winner sequence. 

\begin{theorem}\label{thm:no-elim-tie-bound}
  For any $k \ge 3$, an elimination-tie-free profile must contain at least $(k^3-3k)/2$ voters in order to produce $k-1$ truncation winners.
\end{theorem}

\begin{proof}
Let $x_1 > x_2 > \dots > x_{k-1}>  x_k$ be the first place vote counts sorted in strictly descending order and index candidates in this order. Note that the inequalities must be strict so that eliminations at $h=1$ have no ties. As in the proof of \Cref{thm:tie-bound}, candidate 1 must be overtaken by candidates $2, \dots, k-1$ when candidate $k$ redistributes votes ($h\ge 2$). In order to make candidate $2$ overtake candidate 1 after $k$ is eliminated, $k$ must redistribute at least two ballots to candidate $2$. Similarly, candidate $k$ must redistribute at least $i$ ballots to each candidate $i = 2, \dots, k-1$ for them to overtake candidate $i$. This requires  at least $\sum_{i=2}^{k-1} i= T_{k-1} -  1$ ballots listing $k$ first, where $T_k = k(k+1)/2$ is the $k$th triangular number.  

Candidate $k-1$ thus needs at least $T_{k-1} -  1 + 1$ ballots listing them first since $x_{k-1} > x_k$. Similarly, candidate $i$ needs at least $T_{k-1} -  1 + k - i$ ballots listing them first. Adding up these lower bounds yields the desired lower bound:

\begin{align*}
  \sum_{i = 1} ^k (T_{k-1} -  1 + k - i) &= k(T_{k-1} -  1) + \sum_{i = 1} ^k (k - i)\\
  &= k(T_{k-1} -  1) + T_{k-1}\\
  &= (k+1)(T_{k-1}) - k\\
  &= (k+1)(k-1)k/2 - k\\
  &= (k^3 - 3k)/2.
\end{align*}
\end{proof}

\begin{theorem}\label{thm:no-tie-bound}
  For any $k \ge 3$, a tie-free profile must contain at least $(2k^3 - 5 k^2 + 3k)/2$ voters in order to produce $k-1$ truncation winners.
\end{theorem}
\begin{proof}
The argument is almost the same as in the proof of \Cref{thm:no-elim-tie-bound}, except that when candidate 1 is overtaken by candidates $2, \dots, k-1$, the overtaking candidates cannot be tied afterwards. As before, candidate $k$ needs to distribute at least $k-1$ ballots to candidate $k-1$ to make them overtake candidate $1$. But now, they cannot merely redistribute $k-2$ to candidate $k-2$, since this could cause a tie with candidate $k-1$. In order to make all of $2, \dots, k-1$ overtake candidate 1 and not emerge in a tie, the lowest possible totals $2, \dots, k-1$ could have after reallocation are $x_1+1, x_1+2, \dots, x_1+k-2$, where $x_1$ is the first-round vote total of candidate 1. Thus, the  number of votes candidate $k$ must reallocate is at least $\sum_{i = 1}^{k-2}(x_1 + i) - \sum_{i=1}^{k-2}(x_1 - i)$, where the second sum is an upper bound on the number of votes candidates $2, \dots, k-1$ have in round 1, given that they are all behind candidate 1 and not tied. This allows us to calculate the minimum number of ballots listing $k$ first:

\begin{align*}
  \sum_{i = 1}^{k-2}(x_1 + i) - \sum_{i=1}^{k-2}(x_1 - i) &= 2\sum_{i=1}^{k-2} i\\
  &= (k - 2)(k - 1)
\end{align*}

Candidate $k-1$ thus needs at least $(k - 2)(k - 1) + 1$ ballots listing them first since $x_{k-1} > x_k$. Similarly, candidate $i$ needs at least $(k - 2)(k - 1) + k - i$ ballots listing them first. Adding up these lower bounds yields the desired lower bound:

\begin{align*}
  &\sum_{i = 1} ^k ((k - 2)(k - 1) + k - i) \\
  &= k(k - 2)(k - 1) + \sum_{i = 1} ^k (k - i)\\
  &= k(k - 2)(k - 1) + k(k-1) / 2\\
  &= (2k^3 - 5k^2 + 3k)/2.
\end{align*}
\end{proof}

Note that for consequential-tie-free profiles, the lower bound on voters for $k-1$ truncation winners is $\Omega(k^2)$, but $\Omega(k^3)$ for elimination-tie-free and tie-free profiles.

\begin{theorem}\label{thm:no-tie-construction}
Given the same setup as in \Cref{thm:tie-construction}, there exists a tie-free profile with $(2k^3 - 5 k^2 + 3k)/2$ ballots whose sequence of truncated IRV winners from $h=1, \dots, k-1$ is $w_1, \dots, w_{k-1}$.
\end{theorem}
\begin{proof}
  The construction follows the same idea as in \Cref{thm:tie-construction}, but we no longer have the luxury of maintaining the tie for second place among all candidates who are not about to win or about to be eliminated. Instead, we will maintain gaps of a single vote between candidates, as in our lower bound proof. However, the  order of candidates matters. Given a winner sequence $w_1, \dots, w_{k-1}$, define its $f$-sequence as follows. Let $f_1, \dots, f_\ell$ be the $\ell \le k-1$ distinct truncation winners in the sequence $w_1, \dots, w_{k-1}$ ordered by their first appearance in this sequence. Fill the remainder of the sequence $f_{\ell+1}, \dots, f_k$ in reverse order of full-ballot elimination (i.e., $f_k = 1$), skipping candidates already in $f_1, \dots, f_\ell$. For example, the $w$-sequence $4, 3, 4, 5$ for $k = 5$ would result in the $f$-sequence $4, 3, 5, 2, 1$ (recall that candidates are labeled in order of their full-ballot IRV elimination). Assign ballots to each candidate so that their first place vote counts result in the $f$-sequence, with candidate candidate $f_j$ receiving $(k - 2)(k-1) + k - j$ ballots listing them first. Call the first part of the $f$-sequence the \emph{winner prefix} and the second part the \emph{loser suffix}. We will maintain the following invariant: before step $\ell \le h$ of IRV, the order of the remaining candidates $\ell, \ell+1, \dots, k$ by vote count is the $f$-sequence of $w_\ell, w_{\ell+1}, \dots, w_{k-1}$.
  
 As before, let $S_{i}$ denote the set of ballots listing $i$ first. Except for $i=1$, all ballots in $S_i$ rank candidates $1, \dots, i-1$ in positions $2, \dots, i$. Next, we will fill in position $i+1$ for each $S_i$ to maintain the $f$-sequence invariant. 
  
  \underline{Case (1)} If $w_{i} = w_{i+1}$, all ballots in $S_i$ terminate after position $i$. 
 
  \underline{Case (2)} If $w_i = i + 1$, $k - i$ ballots in $S_i$ list each of candidates $i+2, \dots, k$ in position $i+1$. This requires up to $(k - 2)(k-1)$ ballots. 
  
  \underline{Case (3)} If $w_i$ next wins at ballot length $\ell > i+1$, then we need to insert $w_i$ into this position in the winner prefix. Consider the sequence of winners $w_{i+1}, \dots, w_{\ell - 1}$. Let $w_j$ be the last candidate in this sequence to make their first appearance. We will reallocate votes so that $w_i$ is one vote  behind $w_j$. Let $c$ be the size of the vote gap between $w_j$ and $w_i$ before step $i$ of IRV. For instance, $c=1$ if $w_j = w_{i+1}$. For each candidate starting at $w_{i+1}$ and going down the order of candidates by decreasing vote count before step $i$ of IRV to $w_j$, $c+1$ ballots in $S_i$ list that candidate in position $i+1$. For each candidate starting after $w_j$ in vote count order  and going down to $i + 1$, $c$ ballots in $S_i$ list that candidate in position $i+1$. This requires at most $(k-3)(k-1) < (k - 2)(k-1)$ ballots, an upper bound achieved if $w_j$ has only one more vote than $i+1$ and $i = 1$. 
    
  \underline{Case (4)} If $w_i$ does not appear again in the sequence $w_{i+1}, \dots, w_{k-1}$, then we will insert it into its correct position in the loser suffix. Consider the sequence of subsequent losers $i+1, \dots, k$ and remove candidates that win at truncations lengths $i+1, \dots, k-1$. Let $j$ be the largest-indexed candidate in this pared-down sequence whose index is smaller than $w_i$ (at least one such candidate exists since $i+1$ is eliminated before $w_i$ and can't win at ballot lengths $>i$). We will insert $w_i$ into the loser sequence so that they have one more vote than $j$. Let $c$ be the size of the vote gap between $w_i$ and $j$ before step $i$ of IRV. Consider the order of candidates by vote count before step $i$ of IRV. For each candidate with more votes than $j$ (excluding $w_i$), $c$ ballots in $S_i$ list that candidate in position $i+1$. For each candidate with fewer votes than $j$ (including $j$ but excluding $i$), $c - 1$ ballots in $S_i$ list that candidate in position $i+1$. After reallocation, $w_i$ will then be one vote ahead of $j$ and one vote behind the next candidate above them. This requires at most $(k-3)(k-1)< (k - 2)(k-1)$ ballots, an upper bound achieved if $j=i+1$ and $i =1$.

All ballots terminate after their last specified entry. We now prove that the truncation winner sequence of this profile is $w_1, \dots, w_{k-1}$. We'll prove inductively that the $f$-sequence invariant is maintained by construction.

\underline{Base case} ($\ell = 1$): By construction, the first place vote counts are exactly the $f$-sequence of $w_1, \dots, w_{k-1}$.

\underline{Inductive case} ($\ell\ge 2$): By inductive hypothesis, we have that after step $\ell-1 < h$, the candidates $\ell-1, \dots, k$ were in their $f$-sequence order by decreasing vote count. We also know $\ell-1$ must have been in last place, since they are eliminated $(\ell-1)$st. Consider what occurs when $\ell-1$ is eliminated. We will mirror the four cases of the construction. (1) If $w_{\ell-1} = w_{\ell}$, position $\ell-1$ is empty and their ballots are all exhausted, leaving the order as is. The order of the candidates by vote count remains the $f$-sequence of the remaining candidates. (2) If $w_{\ell-1}= \ell$, then all candidates between $w_{\ell-1}$ and $\ell-1$ overtake $w_\ell$. The new order of candidates is again the $f$-sequence of the remaining candidates, since the winner prefix remains the same starting from $w_\ell$ and $\ell$ moves into last place. (3) If $w_{\ell-1}$ wins again at some $h>\ell-1$ , then our construction places it in the winner prefix exactly where it belongs: in order of first subsequent win. The loser suffix remains unchanged, leaving the correct $f$-sequence. (4)  If $w_{\ell-1}$ does not win again at $h> \ell-1$, then our construction inserts it into the loser suffix where it belongs: just before the highest-indexed non-subsequent-winner with a lower index than $w_{\ell-1}$. Here, the winner prefix in unaffected, leaving the correct $f$-sequence. 

By construction, as soon as a ballot is reallocated, it becomes exhausted. Additionally, just before step $h$ of IRV, all remaining truncated ballots are exhausted. Thus the order remains the same as trailing candidates are eliminated and $w_h$ wins, since they were in the lead at the front of the $f$-sequence before step $h$.  

Finally, this construction uses the number of ballots claimed:

\begin{align*}
  &\sum_{i = 1}^k \left[(k - 2)(k-1) + k - j\right] \\
  &= k(k - 2)(k-1) + k^2 - \sum_{i = 1}^k j\\
  &= k(k - 2)(k-1)+ k^2 - k(k+1)/2\\
  &= (2k^3 - 5 k^2 + 3k)/2.
\end{align*}
\end{proof}

The constructions for consequential-tie-free and tie-free profiles both use $\Theta(k^2)$ \emph{distinct} ballots. However, only $\Theta(k)$ distinct ballots are required to produce $k-1$ truncation winners. This is asymptotically tight, since each candidate who wins at some ballot length needs at least one ballot type listing them first.

\begin{theorem}\label{thm:k-types}
Given $k>3$ candidates, there is a tie-free profile producing $k-1$ truncation winners with $\Theta(k^3)$ voters of $\Theta(k)$ types.
\end{theorem}

\begin{proof}
We'll construct a set of ballots such that candidate $h$ wins at truncation $h = 1, \dots, k-1$. Call the last candidate (the first one eliminated) $k$. Let $x = (2k-4)(k-2)$. Construct $x + 2(k - 1)$ ballots ranking $1$ first and $x+2i$ ballots ranking candidates $i = 2, \dots, k - 2$ first. Construct $x + 3$ ballots ranking candidate $k - 1$ first and $x$ ranking candidate $k$ first. Thus, the order of candidates from most to least first-place votes is $1, k-2, k-3, k-4, \dots, 3, 2, k-1, k$ and the total number of ballots is $\Theta(k^3)$. We now fill in the ballots for each of the candidates.

\underline{Candidate $k$}: Make $2k-4$ of the ballots ranking $k$ first rank each of $2, \dots, k-1$ second. These ballots then terminate. This requires $(2k-4)(k-2)$ ballots of $k - 2$ types.

\underline{Candidate $k-1$}: $2k$ of the ballots ranking $k-1$ first rank candidates $2, \dots, k-2$ in positions $2, \dots, k-2$, then terminate. The remaining ballots ranking $k-1$ first have length 1. Candidate $k-1$ thus uses only two ballot types.

\underline{Candidates $i = 1, \dots, k-2$}: Two of the ballots ranking $i$ first rank candidates $k, 1, 2, \dots, i - 1, k-1$ in positions $2, \dots, i+2$, then terminate---note that candidate $1$'s ballots of this form are $(1, k, k-1)$. The remaining ballots ranking $i$ first have length 1. Candidate $i$ thus uses only two ballot types.

Notice that the construction uses $O(k^2)$ ballots ranking each candidate first, for a total of $O(k^3)$ ballots. These are split among $k-2 + 2 + 2(k-2) = 3k - 2 = \Theta(k)$ types. We now show that truncating ballots at length $h < k$ results in candidate $h$ winning under IRV.  
 
 If $h=1$, candidate 1 wins since they have the most first-place votes. 
 
 If $h=2$, the first candidate eliminated is candidate $k$. Their second-place votes cause candidates $2, \dots, k-1$ to overtake candidate $1$, who is eliminated second. However, candidate $1$'s ballots of the form $(1, k, k-1)$---truncated to $(1, k)$---are now exhausted, so no reallocation occurs. This causes candidate $k-1$ to be eliminated third, and their $2k$ ballots ranking candidate $2$ second cause candidate $2$ to take the lead. The eliminations then proceed in the order $3, \dots, k-2$, with no reallocation since those candidates' ballots are all exhausted. Candidate 2 wins.
 
 For $h > 2$, we'll show inductively that for $\ell = 2, \dots, h-1$ ($h \le k-1$), the $\ell$th candidate eliminated is candidate $\ell -1$, which causes candidate $k-1$ to jump one vote ahead of candidate $\ell$, who falls into last place.
 
 \underline{Base case ($\ell = 2$)}: The first two eliminations proceed as they did for $h=2$. However, when candidate 1 is eliminated, their two ballots ranking $k-1$ third go to candidate $k-1$, since $h>2$. Before this reallocation, candidate $k-1$ was in second-to-last place with $x+3 + 2k-4 = x + 2k - 1$, one vote behind candidate $2$, who had $x + 4 + 2k - 4 = x + 2k$ votes. The reallocation of candidate 1's ballots causes candidate $k-1$ to jump one vote ahead of candidate $2$, who falls into last place.
 
 \underline{Inductive case ($\ell > 2$)}: By inductive hypothesis, candidate $\ell-2$ was last eliminated, which caused candidate $\ell -1$ to drop into last place, one vote behind candidate $k - 1$. Thus, candidate $\ell - 1$ is eliminated next. The next uneliminated candidate listed on their two ballots of length $> 1$ is $k-1$, since candidates $k, 1, \dots \ell-2$ have all been eliminated by inductive hypothesis and the cases above. When candidate $\ell - 1$ is eliminated, those two ballots cause candidate $k-1$ to jump one vote ahead of candidate $\ell$, who has $x + 2\ell$ votes. Since candidate was one vote ahead of candidate $\ell - 1$ (who had $x + 2(\ell - 1)$ votes) and then gained two more, candidate $k -1$ is therefore one vote ahead of candidate $\ell$ after the $\ell$th elimination and redistribution. 
 
 We can now show that for $2 < h < k$, candidate $h$ wins when the ballot length is $h$. Consider such a ballot length $h$. By the inductive argument above, the $(h-1)$th candidate eliminated is candidate $h-2$, which causes candidate $h-1$ to fall into last place, one vote behind candidate $k-1$. Notice that when candidate $h-1$ is eliminated, they do not reallocate any votes to candidate $k-1$, since candidate $k-1$ appears at position $h+1$ on their two nontrivial ballots. Thus candidate $k-1$ is eliminated after candidate $h-1$ (note that if $h=k-1$, candidate $k-1$ wins after $h-1 = k-2$ is eliminated). The $2k$ nontrivial ballots assigned to $k-1$ are then redistributed to candidate $h$, since they are the lowest-indexed non-eliminated candidate and they appear at position $h$ on candidate $k-1$'s nontrivial ballots. These $2k$ ballots are enough to put candidate $h$ in first place. The eliminations then proceed in the order $h+1, h+2, \dots, k-2$, with no reallocation since those candidates' ballots only include candidates indexed lower than them (except $k$ and $k-1$, who have been eliminated). Candidate $h$ then wins.

\end{proof}

\subsection{Full ballots}
So far, all of our constructions have relied on partial ballots. For profiles with full ballots, a simple extension of our constructions using filler candidates allows us to achieve up to $k / 2$ truncation winners, and in fact any feasible sequence of winners in the first half of ballot lengths.

\begin{corollary}\label{cor:full-ballot}
  Let $k = 2\kappa$ for some $\kappa > 3$. Label the candidates $1, \dots, 2\kappa$ in order of their elimination under full ballots. Fix any sequence $w_1, \dots w_{\kappa-1}$ such that $w_h \in \{\kappa + h + 1, \dots, 2\kappa\}$ for all $h \in [\kappa - 1]$. There exists a full-ballot consequential-tie-free profile with $2\kappa^2 - 2\kappa$ voters and a full-ballot tie-free profile with $(2\kappa^3 - 5\kappa^2 + 3\kappa)/2 + \kappa(\kappa-1)/2$ voters whose sequences of truncation winners from $h=1, \dots, \kappa - 1$ are $w_1, \dots, w_{\kappa - 1}$. 
\end{corollary}

\begin{proof}
Perform the same constructions as in the proofs of \Cref{thm:tie-construction,thm:no-tie-construction}, but with $\kappa$ instead of $k$. Then add in another $\kappa$ candidates with zero first place votes, which are always immediately eliminated (for the tie-free construction, these candidates need $0, \dots, \kappa-1$ first-place votes to avoid a tie, resulting in the extra $\kappa(\kappa-1)/2$ ballots). Use them to fill in all the partial ballots generated by the construction up to position $\kappa$. Then fill in all ballots with the remaining candidates arbitrarily. Up to $h=\kappa$, the construction performs exactly as before, since all of the filler candidates are eliminated first regardless of ballot length, allowing the ballots to act as if they were partially filled out. No guarantees are made about the behavior under ballot lengths $h = \kappa +1, \dots, 2\kappa$. 
\end{proof}

 While we have not found a general construction with full ballots and $k-1$ truncation winners, we have found full-ballot elimination-tie-free profiles with $k-1$ truncation winners up to $k=10$ using a linear-programming-based search (described at the end of this section). Full ballots make intuitive constructions more challenging, but do not appear to prevent a large number of truncation winners. However, how a full ballot requirement does or doesn't change our main result remains an open question. 

If instead of requiring ballots to be full, we require them to all have length at least $k/2-c$, we can improve the above extension of our constructions and get an additional $c$ ballot lengths at which we can specify the winner.

\begin{corollary}\label{cor:full-ballot-minus-c}
  Let $k = 2 \kappa$ for some $\kappa > 3$. Suppose we require ballots to have length at least $\kappa - c$ for $c < \kappa$. Label the candidates $1, \dots, 2\kappa$ in order of the elimination under full ballots. Fix any sequence $w_1, \dots w_{\kappa + c -1}$ such that $w_h \in \{\kappa - c + h + 1, \dots, 2\kappa\}$ for all $h \in [\kappa - 1]$. There exists a consequential-tie-free profile with $2\kappa^2 - 2\kappa$ ballots and a tie-free profile with $(2\kappa^3 - 5\kappa^2 + 3\kappa)/2 + (\kappa-c)(\kappa-c-1)/2$ ballots whose sequence of truncated IRV winners from $h=1, \dots, \kappa + c - 1$ is $w_1, \dots, w_{\kappa + c - 1}$.
\end{corollary}

\begin{proof}
  If we can make ballots as short as $\kappa - c$, then we can perform the same construction as above, but using $c$ fewer fillers. Shortening ballots to $\kappa -c$ ensures that ballots will become exhausted when needed, while using fewer than $\kappa - c$ fillers. Those unused fillers are now additional candidates who could be made to win at longer ballot lengths, up to $\kappa + c - 1$. As before, we need to give the fillers $0, \dots, k-c - 1$ first-place votes to avoid ties in the tie-free construction, resulting in the extra $(\kappa-c)(\kappa-c-1)/2$ ballots.  
\end{proof}

Given that explicit full-ballot constructions appear quite challenging, we turn to a computational approach to investigate whether full-ballot profiles can produce $k-1$ truncation winners. Using a linear programming (LP) search, we identified elimination-tie-free profiles with full ballots and $k-1$ truncation winners for $k= 4, 5, 6, 7, 8, 9, 10$ (the sizes of these profiles are shown in \Cref{tab:simplex-constructions}). Moreover, this approach was able to find instances with voter counts matching the exact lower bound in \Cref{thm:no-elim-tie-bound} for $k=5, 6, 7$. Our approach was not able to match the voter lower bound for $k=4$.  For $k\ge 8$, we faced runtime constraints since the number of variables is exponential in the number of candidates, leading us to restrict the search space (described in further detail below). We consider elimination-tie-free profiles since they are easiest to encode as an LP, where we use constraints to enforce unambiguous eliminations.  

\begin{table}[h]
\centering
\caption{LP full-ballot constructions. We used different search strategies for $k\le 7$ and $k \ge 8$, leading to profiles farther from the voter lower bound for $k \ge 8$.}\label{tab:simplex-constructions}
  \begin{tabular}{lrrrr}
\toprule
$k$ & \makecell[r]{\# trunc.\\ winners} &\makecell[r]{ ballot\\ types} & voters & \makecell[r]{voters lower bound\\ (\Cref{thm:no-elim-tie-bound})}\\
\midrule
4 & 3 & 7 & 29 & 26\\
5 & 4 & 12 & 55 & 55\\
6 & 5 & 23 & 99 & 99\\
7 & 6 & 36 & 161 & 161\\
8 & 7 & 57 & 974 & 244\\
9 & 8 & 85 & 1759 & 351\\
10 & 9 & 122 & 4855 & 485\\
\bottomrule
\end{tabular}
\end{table}

The idea behind the search is to construct possible elimination orders across all $h$ that could result in $k-1$ winners, express these as conditions on the sums of counts of every full ballot type in $S_k$ (the set of permutations on $k$ elements), and then use an LP to find a feasible real-valued solution of ballot type counts that result in that elimination order. We round these fractional ballot counts to be integers and check if the resulting profile has the desired elimination order. If not, we can try another possible elimination order or increase the gaps in the constraints so that rounding is less likely to make a solution infeasible. For $k \le 7$, we tested all possible elimination orders, but only tested a single elimination order for $k \ge 8$ due to runtime constraints. 

As an example of constructing elimination orders that result in $k-1$ truncation winners, consider $k=5$. Label the candidates $1, 2, 3, 4, 5$ in order of their full-ballot IRV elimination. If we are to have 4 truncation winners, they must be $2, 3, 4, 5$, and they must win in that order at $h = 1, 2, 3, 4$. By construction, the elimination order at $h=4$ is $1, 2, 3, 4, 5$. However, at $h=3$, the elimination order must be $1, 2, 3, 5, 4$, since 4 wins. At $h=2$, the elimination order can be $1, 2, 4, 5, 3$ or $1, 2, 5, 4, 3$. At $h=1$, it can be one of six options, corresponding to the permutations of $3, 4, 5$: $1, \{3, 4, 5\}, 2$. There are thus 12 possible elimination orders across ballot lengths we to consider that could result in 4 truncation winners. 

Formally, let $x_\pi$ denote the number of ballots with the ranking $\pi\in S_k$ over the candidates. Fix an elimination order over all ballot lengths and let the $(k-1) \times (k-1)$ matrix $E$ store the elimination orders over ballot lengths $1, \dots k-1$. That is, $E_{hi}$ is the index of the candidate eliminated at round $i$ with ballot length $h$. Let $r(h, i)$ denote the set of remaining candidates at round $i$ with ballot length $h$ that are not eliminated at round $i$. Let $b(h, i, j)\subset S_k$ denote the set of ballot types that would be assigned to candidate $j$ at round $i$ with ballot length $h$. Note that $b(h, i, j)$ and $r(h, i)$ depend on the fixed elimination order. Let $C\ge 1$ be the elimination gap (the smallest number of votes by which eliminations are decided).

The linear program for finding $k-1$ truncation winner constructions, minimizing the number of ballots, is then:

\begin{mini*}|l|
  {}{\sum_{\pi \in S_k} x_\pi }{}{}
  \addConstraint{x_\pi}{\ge 0}
  \addConstraint{C + \underbrace{\sum_{\pi \in b(h, i, E_{h i})}x_\pi }_{\text{\# votes for $E_{hi}$}}}{\le \underbrace{\sum_{\pi \in b(h, i, j)} x_\pi}_{\text{\# votes for $j$}}.}
 \end{mini*}
 We have a constraint of the first type for all $\pi \in S_k$ and constraints of the second type for $h = 1, \dots, k-1; \;i = 1, \dots, k-1;\;$ and $j \in r(h, i)$. 
 The second type of constraint encodes each elimination that takes place during IRV at each ballot length, ensuring that the eliminated candidate $E_{hi}$ has fewer votes in that round it is eliminated than each remaining candidate $j \in r(h, i)$. We chose the objective function to find profiles with few ballots. All constructions generated by running the LP are stored in the code and data repository, which contains instructions for viewing them. It remains an open question whether there exist full ballot profiles for every $k$ that result in $k-1$ truncation winners (or any truncation winner sequence). Given the computational evidence from these LPs up to $k=10$ candidates, we conjecture that there are such profiles.

\subsection{Ballot length in simulation}

\begin{figure}[t]
\centering
         \centering
         \includegraphics[width=\columnwidth]{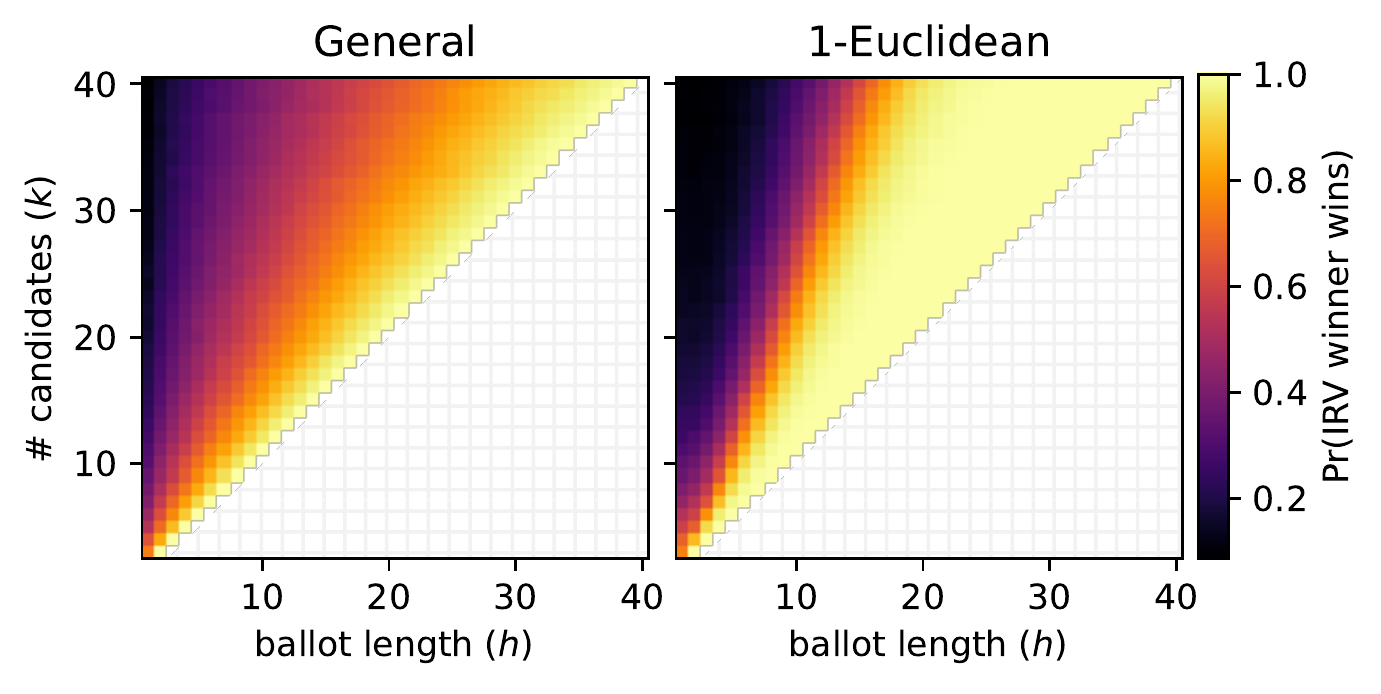}
     
     \caption{Probability that truncated ballots produce the full IRV winner for candidate counts $k=2, \dots, 40$ and ballot lengths $h = 1, \dots, k-1$. (Left) For general ballots (1000, uniform over $S_k$), the probability of producing the IRV winner increases smoothly with the ballot length $h$. (Right) For 1-Euclidean preferences, there is a sharper transition around $h = k/2$. }\label{fig:truncation-heatmaps}
\end{figure}

\begin{figure}[t]
\centering
   \begin{subfigure}[b]{0.8\columnwidth}
         \centering
         \includegraphics[width=\textwidth]{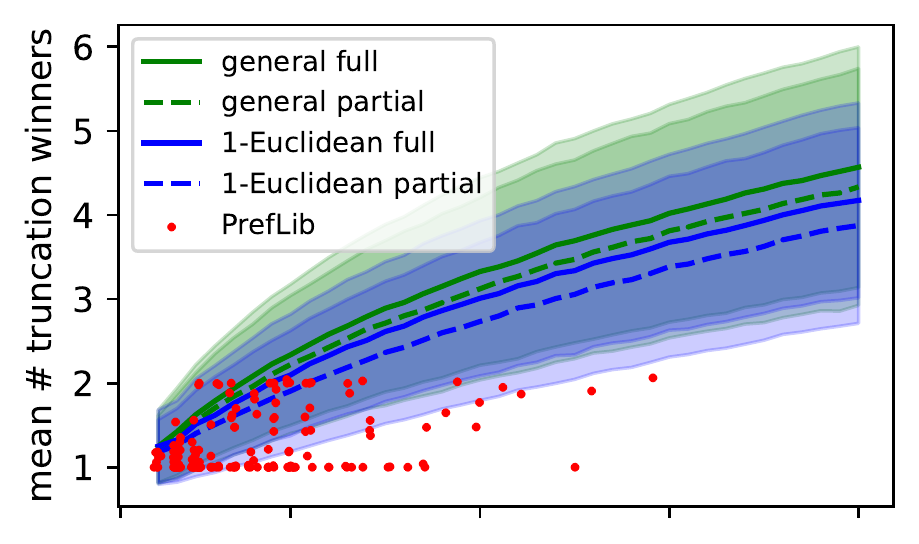}
     \end{subfigure}\\
     \begin{subfigure}[b]{0.8\columnwidth}
         \centering
         \includegraphics[width=\textwidth]{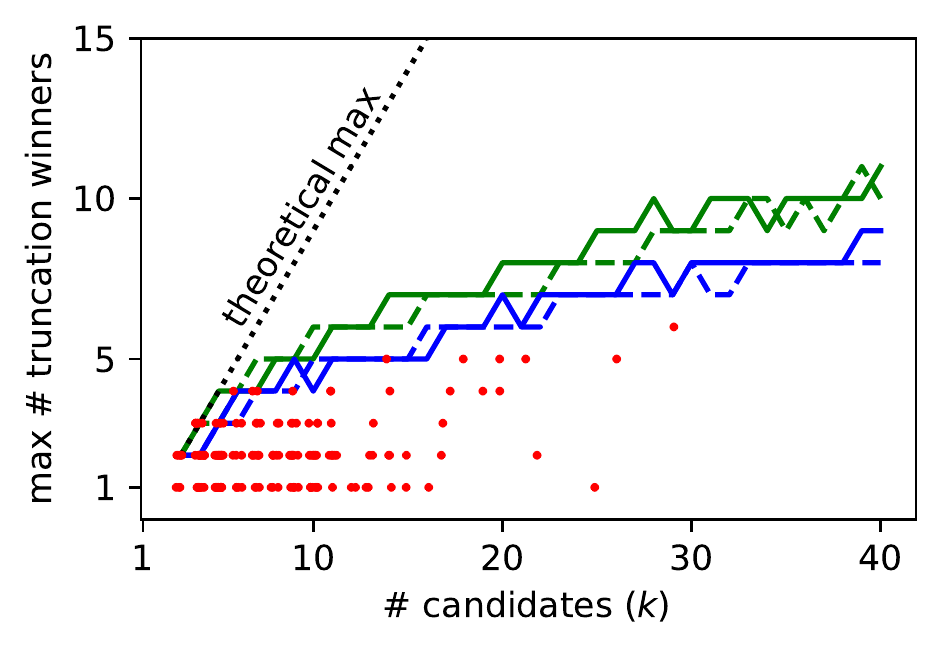}
     \end{subfigure}\\
     
     \caption{Mean (top) and maximum (bottom) number of truncation winners in 10000 synthetic ballot simulations and 10000 PrefLib resampling trials. We simulated uniform general and 1-Euclidean preferences for $k = 3, \dots, 40$.  The shaded regions show standard deviation across trials. To simulate partial ballots, each ballot is voluntarily truncated at a random length between 1 and $k$. While up to $k-1$ truncation winners are possible, the mean number of truncation winners only reaches 4 around $k=40$. 1-Euclidean profiles and profiles with partial ballots tend to produce slightly fewer truncation winners. For the PrefLib data, each point represents a single election, with horizontal jitter added for legibility. Real elections tend to produce even fewer truncation winners, although it is not rare to have more than 1. }
     \label{fig:winner-counts}
\end{figure}

Our theoretical results show that the winner of an IRV election can change dramatically as the ballot length varies. Here, we ask how likely these changes are through simulated profiles. Such simulation analysis was previously conducted for $k=4, 5, 6$~\cite{kilgour2020prevalence}. We extend these simulations up to $k=40$ (our real-world IRV data has examples of elections with up to $\approx 30$ candidates).

We simulate two different types of profiles: general profiles with rankings sampled uniformly at random and 1-Euclidean profiles with voters and candidates embedded in one dimension. For the general profiles, we fix 1000 voters. For 1-Euclidean profiles, we simulate an infinite voter population uniformly distributed over $[0, 1]$, where the number of first-place votes a candidate $i$ has is the size of the interval of $[0, 1]$ containing points closer to $i$ than any other candidate.\footnote{Note that there are at most $\binom{k}{2}+1$ distinct rankings in a 1-Euclidean profile~\cite{coombs1964theory}: the ranking of a voter at $x=0$ lists candidates in left-right order, and as we sweep $x$ to the right, the ranking of a voter at $x$ only changes when we cross one of the $\binom{k}{2}$ midpoints between candidates. We can either compute the regions with these $\binom{k}{2}+1$ distinct rankings or simply iteratively delete the candidate with the fewest first-place votes.}

For both general and 1-Euclidean profiles, we simulate both full and partial ballots to gauge the effect of forced truncation with and without voluntary truncation. For general profiles with partial ballots, we independently and uniformly perform voluntary truncation 
on each voter's preferences before applying forced truncation in the form of ballot length. For 1-Euclidean partial ballots, we do the same with each ballot type.

In \Cref{fig:truncation-heatmaps}, we show the probability that the full-ballot IRV winner is selected with each ballot length $1, \dots, k-1$ for $k=3, \dots, 40$ with initially full ballots (the heatmaps were qualitatively identical for partial ballots; see the Appendix). For general preferences, the probability of selecting the full-ballot IRV winner increases smoothly as ballot length increases. Additionally, for any fixed ballot length, the probability of selecting the IRV winner decreases as the number of candidates increases. For instance, for $h=3$, the probability of selecting the IRV winner first dips below $50\%$ at $k=12$. For 1-Euclidean preferences, small ballot lengths are even less likely to produce the full IRV winner: for $h=3$, the probability first drops below $50\%$ for $k=9$. On the other hand, there is a rapid increase in probability around $h=k/2$. For ballots longer than $k/2$, uniform 1-Euclidean preferences almost always produce the full IRV winner.

In \Cref{fig:winner-counts}, we visualize the same simulation results in a different way. We plot the mean and maximum observed numbers of truncation winners across ballot lengths (the figure also includes PrefLib winner counts described in the next section). While the difference between general and 1-Euclidean profiles was pronounced in the previous heatmaps, they result in almost the same number of truncation winners on average. Additionally, these simulated profiles tend to have a small number of truncation winners relative to the theoretical maximum. On average for $k \le 10$, there are around two truncation winners, while the theoretical maximum is nine. Additionally, the maximum observed number of winners in 10000 simulated trials was well below the theoretical maximum, especially for larger $k$: we only began generating any profiles with 10 truncation winners around $k=40$.

Intuitively, these simulation results therefore indicate that profiles with large numbers of truncation winners are very rare in the space of profiles, at least under these (uniform) measures. However, they do not appear to be significantly rarer among 1-Euclidean profiles than among general profiles, as one might have expected given the increased structure of 1-Euclidean profiles. On the other hand, profiles in which there are more than one winner across ballot lengths are very common. Thus, while truly extreme cases with $k-1$ truncation winners might be rare, cases where ballot length has an effect occur readily in simulation.

\section{Truncating real-world election data}

\begin{table*}[h]
\centering
\caption{Dataset overview.}\label{tab:dataset-summary}
  \begin{tabular}{lrrrrr}
  \toprule
  \bfseries{PrefLib name} & \bfseries{Election locale} & \textbf{\# Elections}  & \textbf{\emph{k}} & \textbf{\emph{h}} & \bfseries{\# Ballots}\\
  \midrule
\texttt{apa} & Am.\ Psych.\ Assoc. & 12 & 5 & 5 & 13318--20239 \\
\texttt{aspen} & Aspen, CO & 2 & 5--11 & 4--9 & 2468--2520 \\
\texttt{berkley} & Berkeley, CA & 1 & 4 & 3 & 4171 \\
\texttt{burlington} & Burlington, VT & 2 & 6 & 5 & 8974--9756 \\
\texttt{debian} & Debian Project & 8 & 4--9 & 4--9 & 143--504 \\
\texttt{ers} & Anon.\ organizations & 87 & 3--29 & 3--29 & 9--3419 \\
\texttt{glasgow} & Glasgow, Scotland & 21 & 8--13 & 8--13 & 5199--12744 \\
\texttt{irish} & Dublin, Ireland & 3 & 9--14 & 9--14 & 29988--64081 \\
\texttt{minneapolis} & Minneapolis, MN & 2 & 7--9 & 3 & 32086--36655 \\
\texttt{oakland} & Oakland, CA & 7 & 4--11 & 3 & 11235--143860 \\
\texttt{pierce} & Pierce County, WA & 4 & 4--7 & 3 & 39974--298438 \\
\texttt{sf} & San Francisco, CA & 14 & 4--25 & 3 & 17675--193854 \\
\texttt{sl} & San Leandro, CA & 3 & 4--7 & 3 & 22360--25316 \\
\texttt{takomapark} & Takoma Park, WA & 1 & 4 & 4 & 202 \\
\texttt{uklabor} & UK Labour Party & 1 & 5 & 5 & 266 \\
\bottomrule
  \end{tabular}
\end{table*}

Given that many truncation winners are theoretically possible, we now ask how often multiple truncation winners occur in real-world election data. To this end, we analyze voter rankings from 168 elections in PrefLib~\cite{mattei2013preflib}. This collection includes 12 American Psychological Association (APA) presidential elections~\cite{regenwetter2007unexpected} ($h=5$), 14 San Francisco local elections ($h=3$), and 21 Glasgow local elections ($h=k$), among others. The number of candidates in these elections ranges from 3--29 and the number of voters from tens to hundreds of thousands (see \Cref{tab:dataset-summary} for an overview). Some of these PrefLib datasets included a small number of ballots with multiple candidates listed at the same rank (0.5\% of all ballots), which we omit.

In order to evaluate the impact of ballot length, we truncate the rankings at each possible shorter ballot length than the election actually used. We then run IRV on the truncated ballots. We assume that if ballots had been shorter, voters would have reported the same ranking, but truncated to the ballot length. It is possible that voters would express their preferences differently depending on the ballot length, so our approach should be seen as an approximation to this counterfactual scenario. 

\begin{figure}[t]
  \centering
         \centering
         \includegraphics[width=\columnwidth]{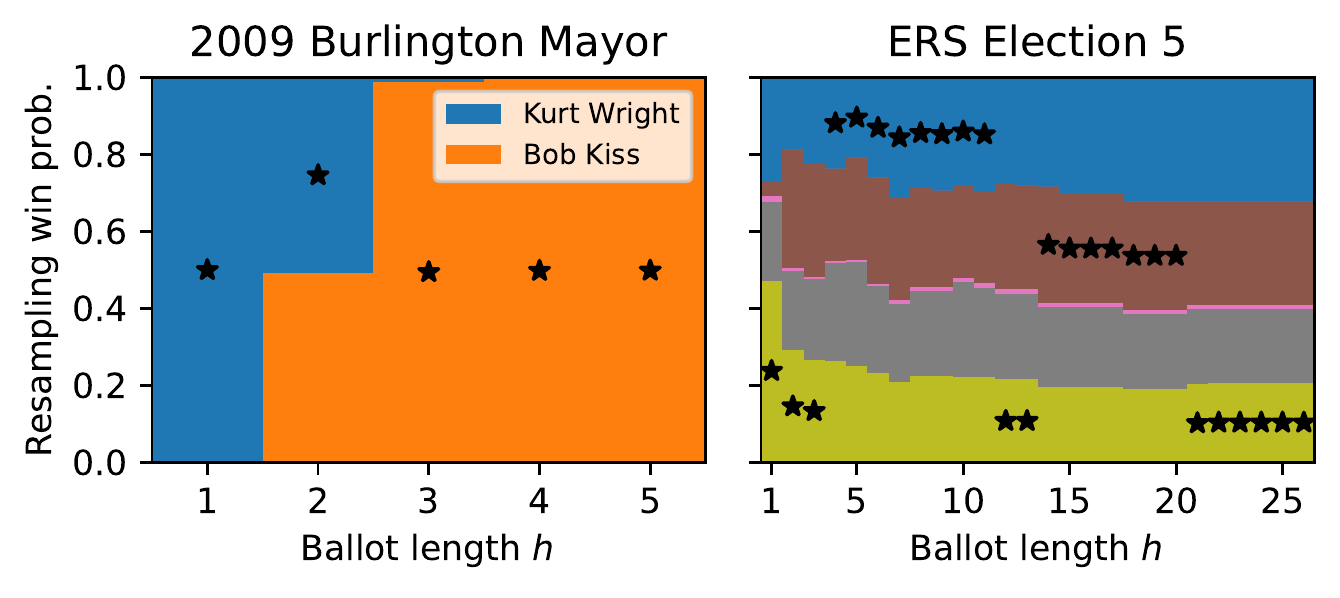}
        \caption{Two elections in the PrefLib data, the infamous 2009 Burlington mayoral election (left, $k=6$, $n=8974$) and an anonymous intra-organization election from the Electoral Reform Society (right, $k=26$, $n=104$). The stacked bars show the probability candidates have of winning at each ballot length under ballot resampling. Stars indicate the winners at each ballot length with actual ballot counts.}
        \label{fig:stacked-bars}
\end{figure}

In 41/168 elections, there were two different winners across ballot lengths, and in one election, there were three different winners. Overall, 25\% of the elections were sensitive to ballot length. Among the elections with ballot length $h \le 5$, $12/85 = 14\%$  of them had two different truncation winners; for elections with $h>5$, $29/83=35\%$ of elections had two or more different winners. In order to better understand the landscape of possible outcomes in each election, we also performed resampling of ballots. Given a collection of $n$ ballots, we resample a collection of $n$ ballots with replacement to simulate another possible election outcome with the same pool of voters. We then truncate those collections of votes to assess the impact of ballot length. In 10000 resampling trials, we observed up to six different truncation winners across the elections, but the expected number of truncation winners under resampling was between one and two for all elections (see \Cref{fig:winner-counts}). In \Cref{fig:stacked-bars}, we also use ballot resampling to visualize the sequence of truncation winners in two PrefLib elections. The 2009 Burlington Mayoral election famously had a different plurality winner (Kurt Wright) than the elected IRV winner (Bob Kiss), but our visualization reveals that at ballot length $h=2$, the election was a complete toss-up and could have gone either way with only a small change in ballot counts. In the right subplot, we visualize the sequence of truncation winners in the one PrefLib election that had three distinct truncation winners. Not only does this election have three truncation winners, but the sequence of winners flips back and forth, as we proved theoretically possible.\footnote{In addition, note that the modal  winner under resampling need not be the actual IRV winner, as we observe in ERS Election 5. A simple example demonstrating this phenomenon is the profile with $k$ candidates and $n$ ballots $(A)$, $n+2$ ballots $(B, A)$, and $n+1$ ballots $(x, A)$ for all other candidates $x$, with $n$ large. The IRV winner is $B$, but $A$ is more likely to win under resampling. For $k=3$, $A$ wins with probability $2/3$ under resampling; as the number of candidates grows, $A$ almost surely wins. This phenomenon is in contrast to plurality, where the actual winner must be the modal resampling winner.}

The smaller number of truncation winners in real data is likely due to the small number of front-runners in real-world elections, in contrast with the uniform preferences in our synthetic data. Our observations here are in line with the finding that ballot truncation is less likely to change the winner in the Mallows model when preferences are more tightly clustered around the central ranking~\cite{ayadi2019single}.

\section{Discussion}

Our theoretical results are fairly pessimistic: IRV election outcomes can change dramatically with ballot length. 
Our analysis of real and simulated data, on the other hand, presents a more mixed picture: ballot length regularly has an effect on the identity of the winner even in real elections, but the extreme changes between winners that are theoretically possible rarely occur, which may be cause for some degree of optimism. Nonetheless, changes in ballot length by truncation can often result in two or three different winners, even when the ballot length is short. 

There are a number of open theoretical questions around ballot length. First, is it possible to achieve every feasible truncation winner sequence with complete ballots? We suspect the answer is yes, but an explicit construction has proved elusive. Second, are more than $O(\sqrt{k})$ truncation winners possible for single-peaked ballots? How many truncation winners are possible with single-crossing ballots? Similar questions could be asked for other profile restrictions, such as 1-Euclidean preferences. 

Our interest in IRV is due to its increasing popularity of IRV in United States local elections, but one could also investigate the effects of ballot length in other ranking-based voting systems such as Borda count or Copeland's method. Additionally, we do not address
what ballot length should be used in practice, which requires making a tradeoff between competing desires. Finally, it would be interesting to understand when elections are close enough for ballot length to affect the winner. There has been research on calculating the margin of victory for IRV~\cite{sarwate2013risk,blom2016efficient,magrino2011computing}, defined as the number of votes which would need to be altered to change the winner, which is NP-hard to compute~\cite{xia2012computing}. A notion of margin of victory that relates to winners across different ballot lengths would be valuable.

\section{Ackowledgments}
This work was supported in part by ARO MURI, a Simons Investigator Award, a Simons Collaboration grant, a grant from the MacArthur Foundation, the Koret Foundation, and NSF CAREER Award \#2143176. We thank the anonymous reviewers for their helpful feedback.

\bibliography{references}

\newpage
\appendix

\section {The IRV Algorithm}
For each voter $j$, let $\pi_j^{(\ell)}$ be their ballot
after step $\ell$ of IRV, with $\pi_j^{(0)} = \pi_j$. 
Let $\pi_j^{(\ell)}(h)$ denote the candidate ranked in position
$h$ by this ballot, with lower indices $h$ corresponding to more 
preferred positions.
A ballot $\pi_j^{(\ell)}$ at step $\ell$ is said to be a \emph{vote} for
candidate $i$ if $\pi_{j}^{(\ell)}(1) = i$. See \Cref{alg:irv} for a formal
definition of the IRV algorithm for determining a winner given a profile
$\{\pi_1, \dots, \pi_n\}$.

\begin{algorithm}[h]
   \caption{Instant runoff voting.}
   \label{alg:irv}
\begin{algorithmic}[1]
    \State {\bfseries Input:} candidates $1, \dots, k$, partial rankings $\pi_j$ over the candidates for each voter $j$
    \State $\pi_j^{(0)} \gets \pi_j, \forall j$
    \State $C = \{i\in \{1, \dots, k\} \mid \exists j: i \in \pi_j\}$\Comment{Non-eliminated candidates}
    \State $\ell \gets 0$
    \While{$|C| > 1$}
        \State $B = \{j \mid |\pi_j^{(\ell)}| > 0\}$\Comment{Non-exhausted ballots}
        \State $i^* \gets \argmin_i \sum_{j \in B} \boldsymbol{1}\left[\pi_j^{(\ell)}(1) = i\right]$\Comment{Break ties as desired}
        \State $\ell \gets \ell + 1$
        \State $C \gets C \setminus \{i^*\}$
        \State $\pi_j^{(\ell)} \gets \pi_j^{(\ell-1)}\setminus \{i^*\}, \forall j$
    \EndWhile
    \State \Return the winner, the last remaining candidate in $C$
\end{algorithmic}
\end{algorithm}


%

\section{Additional figures}\label{app:plots}

\Cref{fig:k-and-h} visualizes the distributions of $k$, $h$, and $n$ in the PrefLib data.

In \Cref{fig:truncation-partial-heatmaps}, we show the versions of the heatmaps in \Cref{fig:truncation-heatmaps} with partial rather than full preferences. For general preferences, we shorted each of the 1000 voters preferences to a length uniform over $1, \dots, k$. For 1-Euclidean voters, we uniformly shorted the preferences of each of the $\binom{k}{2} + 1$ voter types.

\section{Experiment details}
Experiments were run on a server with 144 Intel Xeon Gold 6254 CPUs and 1.5TB RAM running Ubuntu 20.04.4 LTS (Focal Fossa). All libraries used are documented in the code README, as well as detailed instructions for reproducing all experiments.

\begin{figure*}[h]
\centering
  \includegraphics[width=0.9\textwidth]{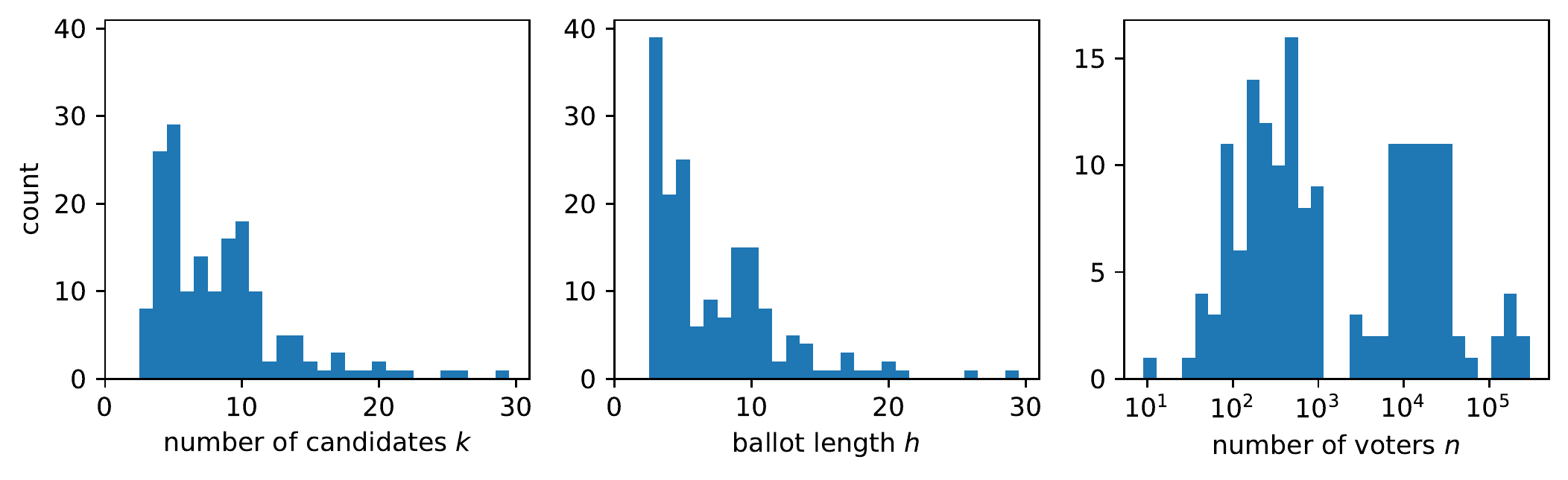}
  \caption{Distributions of candidate counts, ballot lengths, and voter counts in the PrefLib election datasets. }\label{fig:k-and-h}
\end{figure*}

\begin{figure*}[h]
\centering
   \begin{subfigure}[b]{0.4\textwidth}
         \centering
         \includegraphics[width=\textwidth]{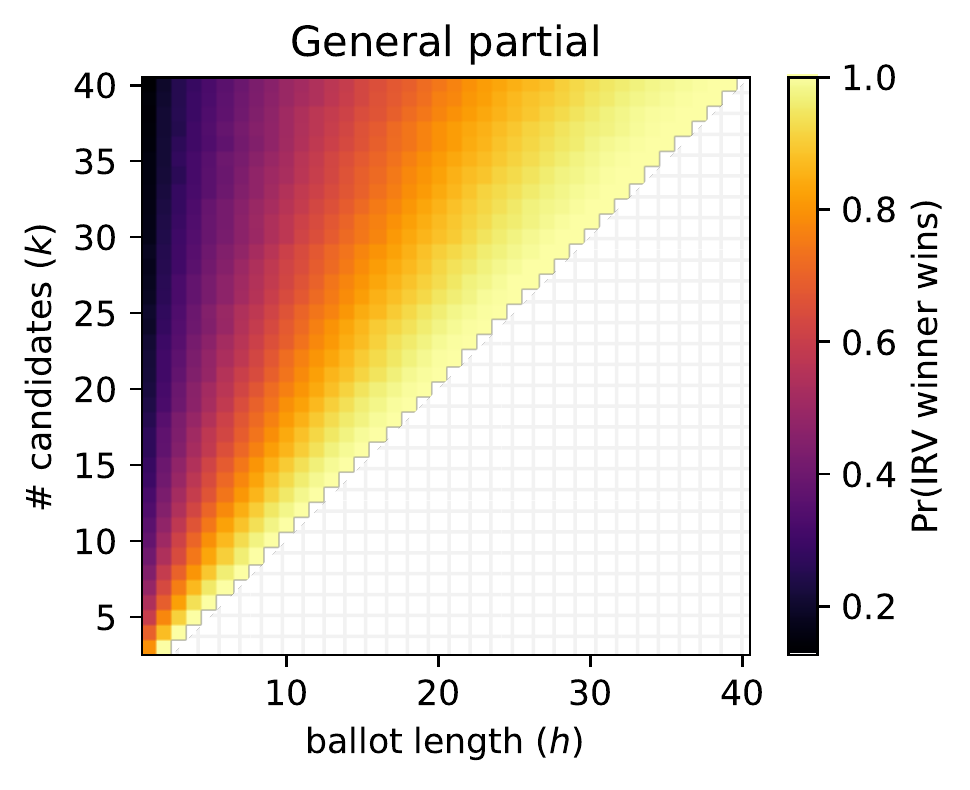}
         \caption{General partial preferences}
     \end{subfigure}
     \qquad
     \begin{subfigure}[b]{0.4\textwidth}
         \centering
         \includegraphics[width=\textwidth]{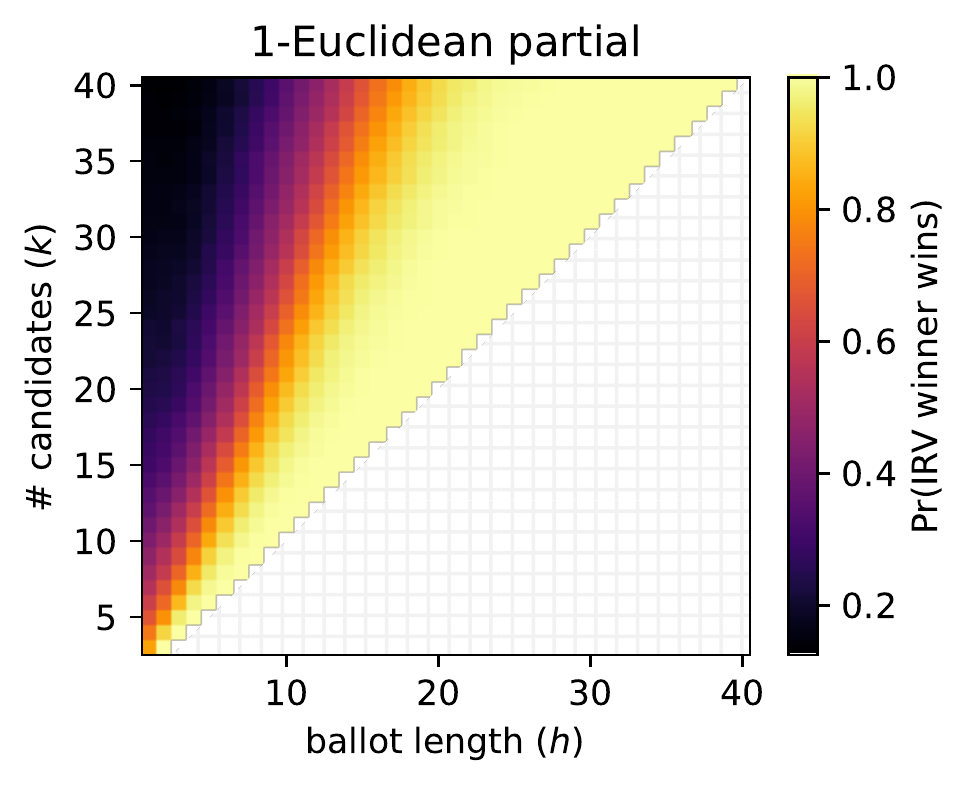}
         \caption{1-Euclidean parital preferences}
     \end{subfigure}\\
     
     \caption{Probability that truncated ballots produce the full IRV winner for candidate counts $k=2, \dots, 40$ and ballot lengths $h = 1, \dots, k-1$ with partial preferences (each voter's preferences are shorted uniformly at random). The results are qualitatively the same as in \Cref{fig:truncation-heatmaps}. }\label{fig:truncation-partial-heatmaps}
\end{figure*}

\end{document}